\documentclass[12pt]{article}
\usepackage{geometry}
\usepackage{amssymb}
\usepackage{amstext,enumerate}
\usepackage{bm}
\usepackage{graphics}
\usepackage{graphicx}
\usepackage{amsthm} 
\usepackage{amsmath}
\usepackage{latexsym}
\usepackage{rotate}
\usepackage{lscape}
\usepackage{color}
\usepackage{epsfig,epsf,psfrag}
\usepackage{sectsty} 
\usepackage{graphics}
\usepackage{epsfig}
\usepackage{multirow}
\usepackage{float}
\usepackage{graphicx}
\usepackage{stmaryrd}
\usepackage{amsmath}
\usepackage{url}
\usepackage{amsthm}
\usepackage{amssymb}
\usepackage{sectsty}
\usepackage{epsfig,natbib}
\usepackage{rotate}
\usepackage{lscape}
\usepackage{setspace}
\usepackage{subcaption}
\usepackage{float}
\usepackage{graphicx}
\usepackage{mathtools}
\usepackage{amstext,amssymb,amsfonts,amscd,color}
\usepackage{graphicx}
\usepackage{epsfig, natbib}
\usepackage{threeparttable}
\usepackage{cases,color}
\usepackage{multirow}
\usepackage{verbatim}
\usepackage{booktabs}
\usepackage{multicol}
\usepackage{bbm}
\usepackage{bm}
\usepackage{subcaption}
\newtheorem{theorem}{Theorem}
\newtheorem{lemma}{Lemma}

\newtheorem{assumption}{Assumption}
\newtheorem{definition}{Definition}

\include{mathdefn2}
\usepackage{mathrsfs}
\usepackage{float}

\usepackage{xparse}
\usepackage{tikz}
\usetikzlibrary{3d}
\usetikzlibrary{plotmarks}
\usetikzlibrary{matrix}
\usetikzlibrary{positioning}
\usepackage{pgfplots}
\pgfplotsset{compat=1.7}
\usetikzlibrary{matrix,backgrounds, arrows.meta}
\pgfdeclarelayer{myback}
\pgfsetlayers{myback,background,main}
\usetikzlibrary{decorations.pathreplacing,angles,quotes}
\tikzset{mycolor/.style = {dashed,rounded corners,line width=1bp,color=#1}}%
\tikzset{myfillcolor/.style = {draw,fill=#1}}%
\tikzset{
	declare function={
		normcdf(\x,\m,\s)=1/(1 + exp(-0.07056*((\x-\m)/\s)^3 - 1.5976*(\x-\m)/\s));
	}
}

\def\blue{\textcolor{blue}}

\sectionfont{\centering\bf\sf\normalsize}
\subsectionfont{\sf\normalsize}

\renewcommand{\baselinestretch}{1.6} 
\newcommand{\single}{\renewcommand{\baselinestretch}{1.2}\normalsize}
\newcommand{\double}{\renewcommand{\baselinestretch}{1.63}\normalsize}

  \renewenvironment{thebibliography}[1]{
    \begin{oldthebibliography}{#1}
      \setlength{\parskip}{0ex}
      \setlength{\itemsep}{0ex}
  }
  {
    \end{oldthebibliography}
  }

\newcommand{\bea}{\begin{eqnarray*}}
\newcommand{\eea}{\end{eqnarray*}}
\newcommand{\be}{\begin{eqnarray}}
\newcommand{\ee}{\end{eqnarray}}
\newcommand{\ed}{\end{document}}

\newcommand{\btab}{\begin{tabular}}
\newcommand{\etab}{\end{tabular}}

\newcommand{\bi}{\begin{itemize}}
\newcommand{\ei}{\end{itemize}}
\newcommand{\bfi}{\begin{figure}}
\newcommand{\efi}{\end{figure}}
\newcommand{\ben}{\begin{enumerate}}
\newcommand{\een}{\end{enumerate}}
\newcommand{\bay}{\begin{array}}
\newcommand{\eay}{\end{array}}

\definecolor{DarkBlue}{rgb}{0,.08,.45}
\definecolor{DarkRed}{rgb}{.7,0,.4}

\def\hg #1 {\texcolor{cyan}{{\it Hans:}   #1}}

\def\H{\mathcal{H}}
\def\S{\mathcal{S}}

\def\bco{\iffalse}

\def\cp{\citep}

\newcommand{\no}{\noindent}
\newcommand{\bc}{\begin{center}}
\newcommand{\ec}{\end{center}}

\newcommand{\bsp}{\begin{split}}
\newcommand{\esp}{\end{split}}
\newcommand{\bdes}{\begin{description}}
\newcommand{\edes}{\end{description}}
\newcommand{\bass}{\begin{assumption}}
\newcommand{\eass}{\end{assumption}}
\newcommand{\bthm}{\begin{theorem}}
\newcommand{\ethm}{\end{theorem}}
\newcommand{\blem}{\begin{lemma}}
\newcommand{\elem}{\end{lemma}}

\def\bco{\iffalse}

\def\cp{\citep}

\DeclareMathOperator*{\argmin}{argmin}

\def\om{\omega}
\def\M{\mathcal{M}}
\def\R{\mathbb{R}}

\bibpunct{(}{)}{;}{a}{}{,}

\begin{document}
\thispagestyle{empty} \single \bc {\bf \sc \Large Geodesic Optimal Transport Regression}
\vspace{0.15in}\\
Changbo Zhu \\
Department of ACMS, University of Notre Dame,
Notre Dame, IN 46556 USA \vspace{0.1in} \\
Hans-Georg M\"uller \\
Department of Statistics, University of California, Davis,
Davis, CA 95616 USA \ec \centerline{16 December 2023}

\vspace{0.1in} \thispagestyle{empty}
\bc{\bf \sf ABSTRACT} \ec \vspace{-.1in} \no 
\setstretch{1} 
Classical regression models do not cover non-Euclidean data that reside in a general metric space, while   the current literature on non-Euclidean regression by and large has focused on scenarios where either predictors or responses are random objects, i.e., non-Euclidean, but not both. 
In this paper we propose geodesic optimal transport regression  models for the case where both   predictors and responses lie in a common  geodesic metric space and predictors may include not only one but also several 
random objects. This provides  an extension of classical multiple regression to the case where  both predictors and responses reside in non-Euclidean metric spaces, a scenario that has not been considered before.   It is based on the concept of optimal geodesic transports, which we define as an extension of the notion of optimal transports in distribution spaces to more general geodesic metric spaces, where we
characterize optimal transports as  transports along geodesics.  The proposed regression models cover the relation between non-Euclidean responses and vectors of non-Euclidean predictors in many spaces of practical statistical interest. 
These include one-dimensional distributions  viewed  as elements of the 2-Wasserstein space and  multidimensional distributions with the Fisher-Rao metric that are represented as data on the Hilbert sphere. Also included are data on finite-dimensional Riemannian manifolds, with an emphasis on spheres, covering  directional and  compositional data, as well as    data  that consist of symmetric positive definite matrices.  We illustrate the utility of geodesic optimal transport regression  with 
data on summer temperature distributions and human mortality.\\

\no {KEY WORDS:\quad Geodesic Metric Spaces; Metric Statistics; Multiple Regression; Random Objects; Ubiquity; Spherical Data; Distributional Data; Symmetric Positive Definite Matrices. 
\thispagestyle{empty} \vfill
\noindent \vspace{-.2cm}\rule{\textwidth}{0.5pt}\\
{\small Research supported in part by  NSF DMS-2310450.}

\newpage
\pagenumbering{arabic} \setcounter{page}{1} \double


\section{INTRODUCTION}

In statistical analysis and modeling one increasingly encounters data that are neither scalars nor vectors or functions. Such complex data can often be considered to be elements of a  non-Euclidean space. Examples  include age-at-death distributions or distributional time series in Wasserstein space \citep{mull:21:9,zhu2021autoregressive},  directional data on the sphere \citep{down, ZHU2023},  diffusion tensor imaging data  in the space of symmetric positive definite matrices \citep{lin2019extrinsic, mull:22:11}, or Corpus Callosum shape contours situated on a  smooth manifold \citep{lin2017extrinsic}. Classical statistical methods are not well suited to deal with such data, as basic arithmetic operations like addition and subtraction are not well-defined due to the absence of a  linear structure. As an example, consider distributional data in the 2-Wasserstein space. Obviously, addition and subtraction are not well defined, because all representations of distributions  come with constraints, e.g., a density function is non-negative and integrates to 1, quantile functions are non-decreasing and constrained at their left and right boundaries, etc., and thus the addition or subtraction of two distributions does not lead to   a valid distribution, irrespective of which representation one adopts. 

To address these challenges, there has been growing interest in extending regression models to accommodate situations where the response or predictor or both are lying in a non-Euclidean space.  Regressing non-Euclidean responses on Euclidean predictors or vice versa has been considered by various researchers, including  \cite{thomas2013geodesic, lin2017extrinsic, lin2019extrinsic, panaretos2022}, among others. A 
recent approach for the case where Euclidean predictors are paired with random objects, i.e., metric-space valued random variables as responses is Fr\'echet regression  \citep{petersen2019frechet}, with recent dimension reduction approaches
\cp{dong:22,virta:22,weng:23,zhan:23} and single-index models  \cp{mull:23:5,ghos:23}. 
However, when both  predictors and response are situated  in a non-Euclidean space, the regression task becomes more complex. Specialized  methods have been devised  by  utilizing the  geometric properties of specific metric  spaces.  For instance, regression models where  both the predictor and response lie on the sphere have been proposed  \citep{down, Michael, ZHU2023}, as well as  distribution-on-distribution regression models that operate in Wasserstein space \citep{ mull:21:9, zhu2021autoregressive,  zhan:22, jiang2022wasserstein, panaretos2022, ghod:23}.

In this paper, we introduce a novel geodesic optimal  transport (GOT) regression model  to accommodate regression tasks where both  predictor and response lie in a common  geodesic metric space, including situations 
where a response is paired with multiple predictors; this scenario is of practical relevance but has not been well studied before, as so far the notion of optimal transports in statistical models was confined to transports of distributions in Wasserstein space.  
The proposed GOT model is based on a   notion of optimal transport in  geodesic spaces that generalizes the classical notion of optimal transport in spaces of probability measures with the Wasserstein metric \cp{mong:81, vill:03, kantorovich2006translocation} and that we introduce in this paper. 

The well-established Wasserstein optimal transports are easily seen to move along geodesics in the metric space of probability measures equipped with the 
Wasserstein metric \cp{mcca:97}. This  provides the motivation for extending  the notion of optimal transport in distribution sopaces to geodesic optimal transport in general geodesic metric spaces. We  say that moving  a random object located in a geodesic space to another random object along the geodesic connecting the two random objects is a geodesic optimal transport.  Geodesic optimal transports are thus  transports along shortest paths, which are well defined as long as the metric space features unique geodesics.
These transports are optimal in the sense of minimizing transport cost under the plausible assumption that  cost is proportional to the length of the transport path.

Fortuitously,  there exists an algebraic structure that can be  used to quantify  movements of objects along geodesics, which was initially introduced for Wasserstein optimal transports in the space of probability distributions \cp{zhu2021autoregressive}
and that we extend here to the more general case of geodesic optimal transports. This transport algebra is a critical ingredient of the proposed GOT regression. 
For all specific geodesic spaces for which we demonstrate the GOT approach we provide explicit constructions for the geodesics and geodesic  optimal transports
that form the basis of GOT regression. 

As mentioned,  GOT regression accommodates multiple predictors. This specifically pertains to the case where the predictors possess no natural known ordering in the same vein as for the classical Euclidean multiple regression, which  provides the inspiration for the proposed GOT regression.   
Since transport operations in general are not commutative, an underlying order needs to be chosen to identify the model. We provide a data-adaptive selector in our methodology so that the proposed approach does not require to pre-specify and order. This feature distinguishes  the GOT model from existing autoregressive distributional models for time series, such as the 
autoregressive optimal transport model (ATM)   for distributional time series in the Wasserstein space \citep{zhu2021autoregressive}, where  the predictors are naturally ordered in  time. We derive prediction consistency 
results for the GOT model under additional geometric constraints on the geodesic space; these are satisfied for the most pertinent spaces of statistical interest.  

The proposed GOT  regression models are illustrated with  human mortality data and temperature data. We compare the performance of GOT regression for these  data with Nadaraya-Watson regression 
\citep{stei:09}, which is also applicable for situations where  both responses and predictors lie in a general geodesic metric space, albeit is subject to the curse of dimensionality which is already severe in the special case of (Euclidean) $L^2$ function spaces due to the well-known small ball probability problem, e.g., for Gaussian processes   \citep{kuel:93};  
 it is potentially a much bigger challenge for data in nonlinear metric spaces; this provides additional motivation for using the transport algebra and the  GOT approach for regression. 

The rest of the paper is organized as follows: In Section \ref{sec:geotrans} we introduce the key concepts of geodesic optimal transport maps and transport algebra and demonstrate in Section \ref{sec:exp}  the construction of transport maps for some specific examples, the Hilbert sphere, 2-Wasserstein space and the space of symmetric positive definite matrices. GOT regression models, estimatres and  supporting theory are developed in Section \ref{sec:gtr}.  Illustrations  for human mortality and summer temperature data are in Section \ref{sec:app}, followed by  Conclusions   in Section \ref{sec:conl}




\section{GEODESIC OPTIMAL TRANSPORTS}\label{sec:geotrans}

Our starting point is a metric space $(\mathcal{M},d)$  with metric $d$ where  the data reside.  
We assume throughout that  $(\mathcal{M},d)$ is bounded and separable. 
 A 
 curve in $\mathcal{M}$ is a  map $r : [a,b] \rightarrow \mathcal{M}$ with length
\begin{align*}
    L(r)  = \sup \left\lbrace \sum_{k=1}^K d( r(t_{k-1}),  r(t_k))  \right\rbrace,
\end{align*}
where the supremum is taken over all $K \in \mathbb{N}$ and partitions $a=t_0 \leq t_1 \leq \dots  \leq t_K =b$ and the curve is  rectificable if $L(r) < \infty$. For two arbitrary points $\omega_1, \omega_2 \in \mathcal{M}$, a connecting path is a curve $r_{\om_1, \om_2}:[a,b] \rightarrow \mathcal{M}$ such that $r_{\om_1, \om_2}(a)=\omega_1$ and $r_{\om_1, \om_2}(b) = \omega_2$. The intrinsic metric $d_I(\omega_1, \omega_2)$ is defined as the infimum of the lengths of all rectifiable paths between any $\omega_1, \omega_2 \in \M$. If $d_I = d$, $(\mathcal{M}, d)$ is called a length space.  A geodesic is a locally length-minimizing path. If for any two points $\om_1, \om_2 \in \mathcal{M}$ there always exists a geodesic $\gamma_{\om_1, \om_2}$ that connects them,  $(\mathcal{M}, d)$ is a geodesic metric space, which is also a length space with $d=d_I$, and if there exists only one such geodesic for all $\omega_1, \omega_2 \in \M$ 
it is a unique geodesic space. We assume throughout  that the random objects are situated in a unique geodesic metric space and 
work with the metric $d=d_I$. The following assumption introduces what we refer to as a  ubiquity property of  the geodesics in  the space $\mathcal{M}$ that is essential for the proposed geodesic transport model. 

\begin{assumption} \label{ass:main} {\it (Ubiquity of geodesics)}. $(\mathcal{M}, d)$ is a unique 
geodesic metric space and there exists a map $\Upsilon: \,  \M \times \M \times \M \rightarrow \M$ such that 
for any $\om_1,\om_2, \om_3 \in \M$, where $\om_1 \ne \om_2$,  there exists a unique element $\om_4\in\M$ such that for any $r \in [0,1]$, 
\begin{align}
\label{eq:main}
    \Upsilon(\om_1, \gamma_{\om_1, \om_2}(r), \om_3) = \gamma_{\om_3, \om_4}(r).
\end{align}
\end{assumption}

Intuitively,  this ubiquity assumption means that any given  geodesic  $\gamma_{\om_1, \om_2}$ can be  
attached  at any element $\om_3 \in \M$, where the new geodesic has the endpoint $\om_4 \in \M$. In other words,  any geodesic 
 $\gamma_{\om_1, \om_2}$ for any $\om_3 \in \M$ can be mapped to a unique geodesic $\gamma_{\om_3,\om_4}$ and we refer to this property as ubiquity of geodesics $\gamma_{\om_1,\om_2}$. 
 The ubiquity map $\Upsilon$  obviously is well-defined in the Euclidean space $\mathbb{R}^p$ where we simply set $\om_4=\om_3 + \om_2 -\om_1$, $ \Upsilon(\om_1, \gamma_{\om_1, \om_2}(r),\om_3)= \om_3 + r(\om_2 - \om_1)$ for any $\om_1, \om_2, \om_3 \in \mathcal{M}$. The same ubiquity construction  works  for Hilbert spaces $L^2([0,1]$ and in fact for any vector space, where geodesics are just vectors that connect the starting and end point. An analogous principle 
 can be applied for Riemannian manifolds through parallel transport, which has been instrumental for various statistical applications \cp{yuan:12,lin:19:2,mull:21:9}.
 
   Nonlinear metric spaces that satisfy this assumption include the space of symmetric positive definite matrices with the Frobenius or power metric, the space of distributions with the 2-Wasserstein metric or the Fisher-Rao metric, compositional data with the square root metric and the space of networks with a fixed number of knots,  equipped with the   Frobenius metric for graph Laplacians \cp{mull:22:11}.   For example for compositional data with the square root metric or distributions with the Fisher-Rao metric, the corresponding spaces  can be viewed as the positive segment of                                                                                                                                                                                                                                                                                                                                                                                                                                                                                                                           a finite-dimensional sphere, respectively,  infinite-dimensional Hilbert sphere with the respective Riemannian geodesics \cp{ZHU2023}; see Section 2.1-2.3 for more details. 
   
   We proceed to introduce the notion of  \emph{geodesic transport}, where we 
   view geodesics $\gamma_{\om_1, \om_2}$ as transport paths that transport objects  $\om_1$ to $\om_2$. More precisely,  
  
   \begin{definition}
Given arbitrary  $\om_1, \om_2 \in \mathcal{M}$, we define the (geodesic) transport $T_{\omega_1, \omega_2}:\mathcal{M} \rightarrow \mathcal{M}$, determined by $\om_1$ and $\om_2$, as
\begin{align*}
T_{\om_1, \om_2}(\om) = \Upsilon (\om_1, \om_2, \om)
\end{align*}
with inverse $T^{-1}_{\om_1, \om_2}: \mathcal{M} \rightarrow \mathcal{M}$ defined as 
$T^{-1}_{\om_1, \om_2}(\om) = \Upsilon (\om_2, \om_1, \om). $
\end{definition}

The idea is that any given geodesic $\gamma_{\om_1,\om_2}$ defines a geodesic transport that can be applied at any $\om_3 \in \M$, transporting $\om_3$  to a uniquely determined element $\om_4 \in \M.$ 
Consider the special case of univariate distributions  with the 2-Wasserstein metric, 
$$
d_{\mathcal W}(\mu_1, \mu_2) = \left\lbrace \int_{0}^1\left( F_{1}^{-1}(x) - F_{2}^{-1}(x) \right)^2 dx \right\rbrace^{1/2},
$$  where $\mu_1, \mu_2$ are two given probability measures with cumulative distributions functions  $F_1, F_2$ that have quantile functions  $F_1^{-1}, F_2^{-1}$, obtained as suitably defined inverses from $F_1, F_2$. 
The geodesic transport of $\om_1$ to $\om_2$  corresponds to the optimal transport of mass distributed according to $\om_1$ to the mass distributed according to $\om_2$ in the sense of Kantorovich \citep{kantorovich2006translocation} and the geodesic that connects $\om_1$ and $\om_2$ is the McCann interpolant \cp{mcca:97}. 

Importantly, for spaces other than Wasserstein space of distributions,  be it the space of distributions in any dimension with the Fisher-Rao metric \citep{dai2022statistical} or a space where the elements are not distributions at all, we can still define  geodesic transports, as long as the ubiquity of geodesics as per Assumption 1 is satisfied. In general geodesic spaces there is no direct interpretation in terms of mass transport, rather the movement of mass is observed as the entire object moves along a geodesic to the target object and this transport is optimal as it proceeds along the shortest path.   For the special case of Riemannian manifolds with the important special case of finite-dimensional spheres or infinite-dimensional Hilbert spheres, to ensure uniqueness of geodesics one needs to remove a part of the manifold  and this  creates a boundary. Geodesic transports  are then constrained to operate within the bounded subset of the Riemannian manifold or sphere, and then need to be appropriately modified so as to never cross the boundary. 

For random objects $X$ taking values in $\mathcal{M}$, and a given probability distribution $\mathcal{P}$ on $\M$,  the Fr\'{e}chet mean $\mu$ is defined as the set 
\begin{align*}
   \mu = \argmin_{\omega \in \mathcal{M}} E[d^2 (\omega,X)],
\end{align*}
where the expectation is with respect to $\mathcal{P}$. We assume in the following that Fr\'echet means are well defined, i.e., exist and are unique, which may depend on the combination of the  distribution $\mathcal{P}$  and  the geometry of $\mathcal{M}$. Uniqueness of the Fr\'echet mean is always guaranteed  for Hadamard spaces \citep{kloeckner2010geometric}, such as the Euclidean space $\mathbb{R}^p$, the Hilbert space $L^2([0,1])$, space of distributions with the 2-Wasserstein metric or the space of symmetric positive definite matrices with the Frobenius metric, to name some prominent examples. But existence and uniqueness are not limited to Hadamard spaces and may also be satisfied in spaces with
positive curvature, depending on the  measure $\mathcal{P}$.

Under Assumption 1, let $T_{\om_1, \om_2}$ be a geodesic transport map corresponding to $\om_1, \om_2 \in \mathcal{M}$ and denote the set of all geodesic transport maps by
\begin{align*}
    \mathcal{T} = \left\lbrace T_{\om_{1}, \om_2} : \om_1, \om_2 \in \mathcal{M} \right\rbrace.
\end{align*}
Following \cite{zhu2021autoregressive}, we  define a scalar multiplication with a factor  $0 < \alpha < 1$ on $T_{\om_1, \om_2}$ by  $\alpha \odot T_{\om_1, \om_2}:\mathcal{M} \rightarrow \mathcal{M}$ such that for any $\om \in \mathcal{M}$
\begin{align}
\label{eq:mult}
    [\alpha \odot T_{\om_1, \om_2}](\om) := \Upsilon (\om_1, \gamma_{\om_1, \om_2}(\alpha), \om), 
\end{align}
For $-1 < \alpha <0$, the scalar multiplication is defined as
\begin{align*}
    \alpha \odot T_{\om_1, \om_2} := |\alpha| \odot T_{\om_1, \om_2}^{-1}.
\end{align*}
For any  $|\alpha| > 1$, let $b = \lfloor |\alpha| \rfloor$, the integer part of $\alpha$, and set $a = |\alpha| - b$.  We define a scalar multiplication by
\begin{align}  \label{mult} 
\alpha \odot T_{\om_1, \om_2} (x) := \left\lbrace  \begin{array}{cl}
(a \odot T_{\om_1, \om_2}) \circ \underbrace{ T_{\om_1, \om_2} \circ T_{\om_1, \om_2} \circ \dots \circ T_{\om_1, \om_2}}_{b \text{ compositions of } T_{\om_1, \om_2}} (x) , &  \alpha > 1  \\
(a \odot T_{\om_1, \om_2}^{-1}) \circ \underbrace{T_{\om_1, \om_2}^{-1} \circ T_{\om_1, \om_2}^{-1} \circ \dots \circ T_{\om_1, \om_2}^{-1} }_{b \text{ compositions of } T_{\om_1, \om_2}^{-1} } (x) , &  \alpha < -1
\end{array} \right. .
\end{align}
Finally, the addition $\oplus$ of two arbitrary geodesic transports $T_1, T_2 \in \mathcal{T}$ is defined as 
$$
T_1 \oplus T_2 := T_1 \circ T_2.
$$

\section{EXAMPLES OF GEODESIC TRANSPORTS} \label{sec:exp}

\subsection{Wasserstein Space}
We provide here examples of geodesic spaces for which the ubiquity assumption (Assumption  \ref{ass:main})  can be satisfied and include  explicit forms for the corresponding geodesic transports. The first example is the Wasserstein space for one-dimensional distributions.
\no Let $\mathcal{W}$ be the set of univariate probability measures on a bounded interval $\mathcal{S} = [s_1, s_2]$ such that $\int_{\mathcal{S}} x^2 d \mu(x) < \infty$ for any $\mu \in \mathcal{W}$. For any measurable function $T:\mathcal{S} \rightarrow \mathcal{S}$, the pushforward measure of $\mu$ by $T$ is denoted as $T_{\#}\mu$. We equip $\mathcal{W}$ with the 2-Wasserstein metric \citep{vill:03}
\begin{align*}
    d_{\mathcal{W}}(\mu_1, \mu_2)  = \inf_{T:T_{\#}\mu_1 = \mu_2} \left\{ \int_{\mathcal{S}} ( T(x) - x )^2 d\mu_1(x) \right\}^{1/2}.
\end{align*} 
The minimizer to the above problem is attained at the optimal transport map 
\begin{align*}
    T_{12} := F_{2}^{-1} \circ F_1,
\end{align*}
where $ F_1=F(\mu_1)$ and  $F_2=F(\mu_2) $ are the cdfs of $ \mu_1 ,\mu_2 $ respectively, and $F_1^{-1}$, $F_{2}^{-2}$ are quantile functions defined by the generalized inverse. Then $(\mathcal{W}, d_{\mathcal{W}})$ is 
a geodesic metric space. For any two distributions $\mu_1, \mu_2 \in \mathcal{W}$, the geodesic between them is given  by McCann's interpolant \citep{mcca:97}, 
\begin{align*}
\gamma_{\mu_1, \mu_2}(a) =\left( id + a ( T_{12} - id ) \right)_{\#} \mu_1, \quad a \in [0,1],
\end{align*}
where $id$ is the identity map. Given $\mu_1, \mu_2 \in \mathcal{M}$ and a probability measure $\mu_3 \in \mathcal{M}$,  Assumption \ref{ass:main} is satisfied with $ \mu_4 = (F_2^{-1} \circ F_1)_{\#} \mu_3 $ and $ \Upsilon (\nu_1, \nu_2, \nu_3) = (G_2^{-1} \circ G_1)_{\#} \nu_3$, where $ \nu_1, \nu_2, \nu_3 \in \mathcal{W}$ are arbitrary and $ G_1, G_2$ are the cdfs of $\nu_1, \nu_2$ respectively. To see this, notice that the corresponding quantile functions of $\mu_4, \gamma_{\mu_1, \mu_2}(a)$ are $F_2^{-1} \circ F_1 \circ F_{3}^{-1}$, $ F_1^{-1} + a(F_2^{-1} -F_1^{-1}) $ respectively, and so
\begin{align*}
    \Upsilon (\mu_1, \gamma_{\mu_1, \mu_2}(a), \mu_3) & = \left( \left\lbrace F_1^{-1} + a(F_2^{-1} -F_1^{-1}) \right\rbrace  \circ F_1 \right)_{\#} \mu_3 \\
    & = \left(  id + a(F_2^{-1} \circ F_1  - id)  \right)_{\#} \mu_3 \\
    & = \gamma_{\mu_3, \mu_4}(a).
\end{align*}


\subsection{Hilbert Sphere}
For  a separable Hilbert space $\mathcal{H}$ with inner product $ \langle \cdot , \cdot \rangle_{\mathcal{H}} $ and norm $\| g \|_{\mathcal{H}} := \sqrt{  \langle g , g \rangle_{\mathcal{H}}}$, for $ g \in \H$ the corresponding Hilbert sphere is  $\mathcal{S} = \{ g \in \mathcal{H} :  \| g \|_{\mathcal{H}} =1  \}$. We consider here the infinite-dimensional Hilbert sphere as well as finite-dimensional versions and use the notation $\S$ to denote both.  In the finite-dimensional case,  if the ambient space is $\R^p$, the sphere is usually denoted by $\S^{p-1}.$ Equipping $\mathcal{S}$ with the intrinsic metric $d(g,h) := \text{arccosin}(\langle g, h \rangle_{\mathcal{H}})$ makes  $(\mathcal{S}, d)$  a geodesic metric space, however the geodesics are not unique. Uniqueness of geodesics can be achieved if we require additional constraints on  $g \in \S$, essentially removing a part of the sphere. 

A statistically meaningful constraint is  $g \ge 0$, which  naturally arises when the elements $g$ of the Hilbert sphere correspond to  square roots of density functions, $g=\sqrt{f}$ in the infinite-dimensional case. Specifically, for the case of (multivariate) absolutely continuous distributions with densities $f$, the Fisher-Rao metric pertains to the geodesic distances between square-root densities and for densities $f_1, f_2$ is defined as 
\begin{align}
\label{fr}
d_{\text{FR}}(f_1, f_2)= \text{arccosin}(\langle\sqrt{f_1}, \sqrt{f_2}\rangle).
\end{align} 
Analogously, compositional data, i.e., vectors with non-negative elements that sum to 1 and appear in data that correspond to percentages and proportions,   can be viewed as elements of a 
finite-dimensional Hilbert sphere when taking their square roots \cp{scea:11} and the metric is analogous to the Fisher-Rao metric. 
  

We  show that Assumption \ref{ass:main} can be enforced
 by explicitly constructing the map $\Upsilon$. For any $g_1,g_2 \in \S$,  
the first step is to  define a rotation operator $\mathbf{R}_{g_1, g_2}(\vartheta)$ that rotates the sphere counterclockwise within $\emph{\text{span}}\{ g_1, g_2 \}$ by an angle $\vartheta$, 
\begin{align} \label{eq:rotation}
\mathbf{R}_{g_1, g_2}(\vartheta) := \exp( \vartheta Q_{g_1, g_2}) = I + \text{sin}(\vartheta) Q_{g_1, g_2} + (1 - \text{cos}(\vartheta))Q_{g_1, g_2}^2,
\end{align}
where $I$ is the identity operator and with $ u_1 = g_1 \text{ and } u_2 = (g_2 - \langle g_2, u_1 \rangle u_1) / \| g_2 - \langle g_2, u_1 \rangle u_1 \|_{\mathcal{H}}$ as the orthonormalized version of $g_1, g_2$, we define 
$Q_{g_1, g_2}:= u_1 \otimes u_2 - u_2 \otimes u_1$.

 Note that the geodesic $r_{g_1, g_2}:[0,1] \rightarrow \mathcal{M}$ between $ g_1 $ and $g_2$ can be traced using $\mathbf{R}_{g_1, g_2}(\vartheta)$ as $ r_{g_1, g_2}(a) = [\mathbf{R}_{g_1, g_2}( a \theta)](g_1) $, where $\theta = \text{arccosin}(\langle g_1, g_2 \rangle_{\mathcal{H}})$. Then the map $\Upsilon$ that satisfies Assumption \ref{ass:main} can be defined as
\begin{align*}
    \Upsilon(g_1, r_{g_1, g_2}(a), g_3) :=  [\mathbf{R}_{g_3, g_4}(a \widetilde{\theta})](g_3),
\end{align*}
where $\widetilde{\theta} = \text{arccosin}(\langle g_3, g_4 \rangle_{\mathcal{H}}) $ and $g_4 := [\mathbf{R}_{g_1, g_2}(\theta)](g_3)$ for $ \theta = \text{arccosin}(\langle g_1, g_2 \rangle_{\mathcal{H}}) $. To satisfy additional constraints and construct a map $\Upsilon:\mathcal{D} \times \mathcal{D} \times \mathcal{D} \rightarrow \mathcal{D}$ that operates in some convex subset $\mathcal{D} \subset \mathcal{S}$, we utilize a user-specified projection operator $\text{Proj}_{\mathcal{D}} : \mathcal{M} \rightarrow \mathcal{D}$ and define $g_4 := \text{Proj}_{\mathcal{D}}( [\mathbf{R}_{g_1, g_2}(\theta)](g_3))$; this modified map $\Upsilon$  satisfies Assumption \ref{ass:main}.

Pertinent applications include the analysis of directional data, which can be represented as a three-dimensional vector $(x,y,z) \in \mathbb{R}^3$ lying on a sphere, i.e., satisfying  $x^2+y^2+z^2 = 1$. We note  that the set of rotation matrices is also known as  the special orthogonal group (SO(3)) and is a group with the group operation defined as $\mathbf{R}_1 \cdot \mathbf{R}_2 = \mathbf{R}_1 \mathbf{R}_2$.

\subsection{Space of Symmetric Positive Definite Matrices}


We denote by $\mathcal{S}_m^{+}$ be the set of $m \times m $ symmetric positive-definite (SPD)  matrices. For any $ S \in  \mathcal{S}_m^{+}$, there exists a unique lower triangular matrix $L$ with positive diagonal entries such that $S = L L^T$ and  a diffeomorphism $\mathscr{L}$ between $\mathcal{S}_m^{+}$ and $\mathcal{L}_{+}$, the set of lower triangular matrices with positive diagonal, where for a given SPD matrix the 
unique lower triangular matrix $\mathscr{L}(S)$ is the Cholesky factor of $S$. For any $L \in \mathcal{L}_{+}$, let $\lfloor L \rfloor$ be the non-diagonal part and  $\mathscr{D}(L)$ the diagonal part,  respectively. By equipping  $\mathcal{L}_{+}$ with the  metric
\begin{align*}
    d_{\mathcal{L}_{+}} (L_1, L_2) := \left\lbrace \|\lfloor L_1 \rfloor - \lfloor L_2 \rfloor \|_{F}^2  +  \| \mathscr{D}(L_1) - \mathscr{D}(L_2) \|_F^2 \right\rbrace^{1/2},
\end{align*}
where $ \| \cdot \|_F $ is the Frobenius norm, $(\mathcal{L}_{+}, d_{\mathcal{L}_{+}})$ becomes a geodesic metric space. Following \cite{lin2019riemannian}, the geodesic w.r.t $d_{\mathcal{L}_{+}}$ between $L_1$ and $L_2$ is   
\begin{align*}
    \gamma_{L_1, L_2}(a) = \lfloor L_1 \rfloor + a( \lfloor L_2 \rfloor - \lfloor L_1 \rfloor ) + \text{exp} \left\lbrace \text{log}(\mathscr{D}(L_1)) + a(\text{log}(\mathscr{D}(L_2)) - \text{log}(\mathscr{D}(L_1))) \right\rbrace.
\end{align*}
Given any  $L_1, L_2, L_3 \in \mathcal{L}_{+}$, it is easy to show that Assumption \ref{ass:main} holds with $L_4$ such that 
\begin{align*}
   & \lfloor L_4 \rfloor = \lfloor L_3 \rfloor + ( \lfloor L_2 \rfloor - \lfloor L_1 \rfloor ), \\
   & \mathscr{D}(L_4) = \text{exp} \left\lbrace \text{log}(\mathscr{D}(L_3)) + (\text{log}(\mathscr{D}(L_2)) - \text{log}(\mathscr{D}(L_1))) \right\rbrace
\end{align*}
and $\Upsilon : \mathcal{L}_{+} \times \mathcal{L}_{+} \times \mathcal{L}_{+} \rightarrow \mathcal{L}_{+}$ defined as $\Upsilon(M_1, M_2, M_3) = M_4$, where 
$ \lfloor M_4 \rfloor = \lfloor M_3 \rfloor + ( \lfloor M_2 \rfloor - \lfloor M_1 \rfloor ) $ and $ \mathscr{D}(M_4) = \text{exp} \left\lbrace \text{log}(\mathscr{D}(M_3)) + (\text{log}(\mathscr{D}(M_2)) - \text{log}(\mathscr{D}(M_1))) \right\rbrace $ for arbitrary $M_1, M_2, M_3 \in \mathcal{L}_{+}$. 

We note that geodesic optimal transports that satisfy the ubiquity assumption can also be constructed for various other metrics in SPD space, as well as for the space of graph Laplacians which correspond to networks when viewed as random objects \cp{mull:22:11}.

\section{GEODESIC OPTIMAL TRANSPORT REGRESSION}\label{sec:gtr}
Suppose we observe a sample of $n$ paired points $\{(\bm{X}_{i}, Y_i): i=1,2, \cdots, n \} \sim^{i.i.d} (\bm{X}, Y)$, where  $ \bm{X} = (X_1, X_2, \dots, X_p)  \in \mathcal{M}^p$ and  $Y \in \mathcal{M}$ where $(\bm{X}, Y) $ have a joint distribution in terms of a  probability measure  on the space $\M^{p+1}.$ Our goal is to develop a regression model that allows to regress $Y_i$ on the $p$ metric-space valued predictors $\bm{X}_i = (X_{i, 1}, X_{i, 2} \cdots, X_{i, p})$ in a principled way. 

 Motivated by the multiple linear regression model in Euclidean spaces (see below for more details on this), we  introduce the following geodesic optimal transport (GOT) regression  model, 
\begin{align}
\begin{split}
    Y_i & = \varepsilon_i \oplus \alpha_1 \odot T_{i, j_1^{*}} \oplus \alpha_{ 2 } \odot T_{i, j^{*}_{2} } \oplus \cdots \oplus  \alpha_p \odot T_{i, j^{*}_p }  (\nu), \label{got}
\end{split}
\end{align}
where $\bm{\alpha} = (\alpha_1, \dots, \alpha_p)^T$ are the true model parameters; $\bm{\mu} = (\mu_1, \dots, \mu_p) \in \mathcal{M}^p$, where  $\mu_j$ is the Fr\'{e}chet mean of $X_j$; $\nu$ is the Fr\'{e}chet mean of $Y$;  $ T_{i,j}:= T_{\mu_{j}, X_{i, j} } \in \mathcal{T} $ is the geodesic transport map from $\mu_{j}$ to $X_{i, j}$; $j_1^{*}, j_2^{*}, \dots, j_p^{*}$ is the true ordering of indices $1, 2, \dots, p$ of the $p$ predictors, where we assume that such a true ordering of the predictors exists, but is unknown;  lastly, $\{\varepsilon_i:\mathcal{M} \rightarrow \mathcal{M} \}_{i=1}^n$ are random perturbation maps \cp{mull:22:8}, i.e., i.i.d random geodesic transports that take values in $\mathcal{T}$ such that for any fixed $z \in \mathcal{M}$,
\begin{align}
\label{eq:epsilon}
    z = \argmin_{ \om \in \mathcal{M}} E\left[  d^2 (\om, \varepsilon_i(z))  \right].
\end{align}
These random perturbation maps are the equivalent in metric spaces of i.i.d. zero mean additive errors in Euclidean spaces.  

We use the notations  $\bigoplus_{k=1}^p \alpha_{k} \odot T_{i, j_{k}^{*}}   := \alpha_1 \odot T_{i, j_1^{*}} \oplus \alpha_{ 2 } \odot T_{i, j^{*}_{2} } \oplus \cdots \oplus  \alpha_p \odot T_{i, j^{*}_p }$ and $T_j = T_{\mu_j, X_j}$ and refer to  $E[d^2 (Y, Z)]$ as the prediction error of a random object $Z \in \M$, when $Z$ is viewed as a predictor of the response $Y$.   The following assumption is needed to make it possible to identify from the observed data the assumed-to-exist latent ordering of the predictors $j_1^{*}, \dots, j_p^{*}$, which is a permutation of  $1,2, \dots.$
\begin{assumption}\label{ass:order}
    The ordered predictors   $X_{j_1^{*}}, X_{j_2^{*}}, \dots, X_{j_p^{*}}$ have the following property.  
    \begin{enumerate}
        \item $X_{j_1^{*}}$ minimizes the prediction error for $Y$ among all single predictors $X_1, \dots, X_p$ in the sense that  
        \begin{align*}
             \min_{\alpha_1} E[d^2 (Y, \alpha_1 \odot T_{ j_1^{*} }) ] & < \min_{j \neq j_1^{*}} \min_{\alpha_1} E[d^2 (Y, \alpha_1 \odot T_{j})].
        \end{align*}
        \item $X_{j_2^{*}}$ minimizes the remaining prediction error for $Y$ among all remaining single predictors after $X_{j_1^{*}}$ has been selected in model \eqref{got} as the first predictor,
        \begin{align*}
        \min_{\alpha_1, \alpha_2} E[d^2(Y, \alpha_1 \odot T_{j_1^{*}} \oplus \alpha_2  \odot T_{j_2^{*}}  ) ] & < \min_{j \neq j_1^{*}, j_2^{*}} \min_{\alpha_1, \alpha_2} E[d^2(Y, \alpha_1 \odot T_{j_1^{*}} \oplus \alpha_2 \odot T_{j})],  
        \end{align*}
        \item In general, for $m<p$, $X_{j_m^{*}}$, minimizes the remaining prediction error after $X_{j_1^{*}}, \dots, X_{j_{m-1}^{*}}$ have been selected as the first $m-1$ predictors in model \eqref{got}, i.e.,  
        \begin{multline*}
    \min_{\alpha_1, \dots, \alpha_{m}} E\left[d^2 \left(Y, \bigoplus_{k=1}^{m} \alpha_{k} \odot T_{ j_{k}^{*}} (\nu) \right) \right]  < \\
\min_{j \neq j_1^{*}, \dots,  j_{m}^{*}} \min_{\alpha_1, \dots,  \alpha_{m}}  E\left[d^2 \left(Y, \bigoplus_{k=1}^{m-1} \alpha_{k} \odot T_{ j_{k}^{*}} \oplus \alpha_{m} \odot T_{j} \right)\right].
    \end{multline*}
    \end{enumerate}
\end{assumption}

 {\it Motivation of the Geodesic Optimal Transport Model \eqref{got}.} As already mentioned, in the Euclidean case  $\mathcal{M} = \mathbb{R}^d$, geodesic optimal transport maps $T$ are simply  differences, $T_{X_1, X_2} = X_2 - X_1$ for any $X_1, X_2 \in \mathbb{R}^d$, while the  geodesic optimal transport regression model reduces to
\begin{align}
    Y_i = \epsilon_i + \sum_{k=1}^p  \alpha_j (X_{i, j_{k}^{*}} - E[X_{i,j_k^{*}}] )  + E[Y_i]. \label{eucl} 
   \end{align}
and thus to the classical multiple linear regression model in $\R^d$ with $p$ predictors, where $d=1$ is the most common case that features scalar responses. Thus the proposed GOT regression emerges as an extension of classical linear regression when data are situated in geodesic spaces. 
Since the addition operation in Euclidean spaces including Hilbert spaces is commutative, the ordering of the predictors is irrelevant in the Euclidean linear case and Assumption 2 is only needed for non-commutative scenarios, as encountered in the transport model  \eqref{got}.

\section{MODEL FITTING AND PREDICTION CONSISTENCY}

We discuss here the estimation of the GOT regression model parameters. Since the  true Fr\'{e}chet means $\bm{\mu} = (\mu_1, \dots,  \mu_p)$ and $\nu$ are generally unknown, we replace them in Model \eqref{eq:main} with  consistent estimates $\widehat{\bm{\mu}} = (\widehat{\mu}_1, \dots, \widehat{\mu}_p)$ and $\widehat{\nu}$, where
 \begin{align*}
     \widehat{\mu}_j  = \argmin_{\om \in \mathcal{M}}  \frac{1}{n} \sum_{i=1}^n d^2(\om, X_{i, j}) \text{ and }
     \widehat{\nu}  = \argmin_{\om \in \mathcal{M}} \frac{1}{n} \sum_{i=1}^n d^2(\om, Y_{i}).
 \end{align*}


The following assumption is from \cite{dubey2019frechet} and  ensures that  $d(\widehat{\mu}, \mu) = o_p(1)$, $ d(\widehat{\nu}, \nu) = o_p(1).$ Throughout, we study 
convergence   as the sample size $n \rightarrow \infty.$ 

\begin{assumption} \label{ass:mean} $\mathcal{M}$ is bounded, $\mu, \nu$ are unqiue and there exist $\xi >0 ,C >0$ such that
\begin{align*}
    & \inf_{d(\omega, \mu) < \xi} \left\lbrace E( d^2(X, \omega)) - E(d^2(X, \mu)) - Cd^2(\omega, \mu)  \right\rbrace \geq 0, \\
    & \inf_{d(\omega, \nu) < \xi} \left\lbrace E( d^2(Y, \omega)) - E(d^2(Y, \nu)) - Cd^2(\omega, \nu)  \right\rbrace \geq 0.
\end{align*}
\end{assumption} 
The geodesic optimal  transports are then constructed based on $\widehat{\mu}_j$ and $\widehat{\nu}$. Setting $\widehat{T}_{i, j} = T_{\widehat{\mu}_j , X_{i,j}}$, the true ordering $j_1^{*}, j_2^{*}, \dots, j_p^{*}$ can be estimated sequentially by 

\begin{align}
    \widehat{j}_1 & = \argmin_{ j } \min_{\alpha_1} \frac{1}{n} \sum_{i=1}^n d^2(Y_i, \alpha_1 \odot \widehat{T}_{i, j} (\nu) ),  \nonumber \\
     \widehat{j}_2 & = \argmin_{j \neq \widehat{j}_1 } \min_{\alpha_1, \alpha_2} \frac{1}{n} \sum_{i=1}^n d^2(Y_i, \alpha_1 \odot \widehat{T}_{i, \widehat{j}_1 } \oplus \alpha_2 \odot \widehat{T}_{i, j}  (\nu) ), \nonumber  \\
    &  \vdots  \label{ord} \\
    \widehat{j}_{p-1} & = \argmin_{j \neq \widehat{j}_1, \dots, \widehat{j}_{p-2} } \min_{\alpha_1, \dots, \alpha_{p-1}} \frac{1}{n} \sum_{i=1}^n d^2 \left(Y_i,  \bigoplus_{k=1}^{p-2} \alpha_{k} \odot \widehat{T}_{i, \widehat{j}_{k}} \oplus \alpha_{p-1} \odot \widehat{T}_{i, j } (\nu) \right) \nonumber
\end{align}
and $\widehat{j}_p$ is  determined as the left-over index  $\widehat{j}_p \neq \widehat{j}_1, \dots, \widehat{j}_{p-1}$. The following assumption, characterizing the geometric properties of $\mathcal{M}$, is needed to show the convergence of $\widehat{j}_1, \dots, \widehat{j}_p$.

\begin{assumption}\label{ass:geometry}
\begin{itemize}
\item[(A)] There exists a constant $C_{\gamma}>0$ such that for any $\om_1, \om_2, \om_3 \in \mathcal{M}$ and $\alpha \in [0,1]$,
\begin{align*}
 d(\gamma_{\om_1, \om_3}(\alpha), \gamma_{\om_2, \om_3}(\alpha) ) \leq C_{\gamma} d( \om_1, \om_2 ).
\end{align*}
\item[(B)] There exists a constant $C_{\Upsilon} > 0$ such that
\begin{align*}
d (\Upsilon(\om_1, \om_2, \om_3), \Upsilon(\om_1', \om_2', \om_3')) & \leq C_{\Upsilon} \left\lbrace d(\om_1, \om_1') + d(\om_2, \om_2') + d(\om_3, \om_3') \right\rbrace.
\end{align*}
\end{itemize}
\end{assumption}

It is not difficult to see that Assumption \ref{ass:geometry} is satisfied if $\mathcal{M}$ is a Hilbert space, a CAT(0) space or an orthant of a Hilbert sphere. 
Wasserstein space and the space of symmetric positive definite matrices are  CAT(0) spaces. The temperature data presented in Section \ref{app:temperature} are considered to be distributional data with the Fisher-Rao metric and as such are supported on the positive orthant of a Hilbert sphere 
and Assumption  \ref{ass:geometry} is thus satisfied.

The following result establishes the consistency of the data-based order selection as per \eqref{ord}.

\begin{theorem} \label{thm:order} Under Assumptions \ref{ass:main} -- 
\ref{ass:geometry},  as $n \rightarrow \infty$, 
    \begin{align*}
        P(\widehat{j}_k = j_k^{*}\,\,\,  {\rm  for } \,\, k =1,2, \dots, p) \rightarrow 1.
    \end{align*}
\end{theorem}
After having demonstrated consistency   of the estimated order of the predictors, our next goal is to establish prediction consistency. Specifically, we aim to find a set of parameters $\bm{\alpha}$ in a compact set $\mathbb{K}$ that is as close as possible to the set of prediction-minimizing parameters $\Gamma \subseteq \mathbb{K}$. Here for a parameter vector $\gamma$ the fact that $\gamma \in  \Gamma$ is equivalent to 
\begin{align*}
\gamma= \argmin_{\bm{\beta} \in \mathbb{K}} 
E\left[d^2 \left(Y, \bigoplus_{k=1}^{p} \beta_{k} \odot T_{ j_{k}^{*}} (\nu) \right) \right],
\end{align*}
i.e., $\Gamma$ is the set of parameter vectors  that minimize  prediction error.
Substituting empirical distributions leads to the following estimators,
\begin{align}
\widehat{\bm{\alpha}} = (\widehat{\alpha}_1, \widehat{\alpha}_2, \dots, \widehat{\alpha}_p)& = \argmin_{\bm{\beta} \in \mathbb{K} } \frac{1}{n} \sum_{i=1}^n d^2 \left(Y_i,  \bigoplus_{k=1}^{p} \beta_{k} \odot \widehat{T}_{i, \widehat{j}_{k}} (\nu) \right), \label{empest}
 \end{align}
where we select an arbitrary representative vector 
among the minimizers of the r.h.s. of
\eqref{empest} as $\widehat{\bm{\alpha}}$  if there is more than one minimizer.

We  quantify the discrepancy between 
$\widehat{\bm{\alpha}}$ and  the set of prediction-minimizing parameters $\Gamma$ by comparing squared prediction errors,
\begin{align*}
  \Delta (\widehat{\bm{\alpha}}, \Gamma ) = E\left[d^2 \left(Y, \bigoplus_{k=1}^{p} \widehat{\alpha}_{k} \odot T_{ j_{k}^{*}} (\nu) \right) \right] - \min_{\bm{\alpha} \in \mathbb{K} } E\left[d^2 \left(Y, \bigoplus_{k=1}^{p} \alpha_{k} \odot T_{ j_{k}^{*}} (\nu) \right) \right],
\end{align*}
where $Y = \varepsilon \oplus \bigoplus_{k=1}^{p} \alpha_{k}^{*} \odot T_{ j_{k}^{*}} (\nu)$ for some $\bm{\alpha}^{*} = (\alpha_1^{*}, \dots, \alpha_p^{*}) \in \Gamma$. 
Based on the definition of $\varepsilon$ in Equation \eqref{eq:epsilon}, it can be seen that $\Delta(\widehat{\bm{\alpha}}, \Gamma) \geq 0$ and $ \Delta(\widehat{\bm{\alpha}}, \Gamma) = 0 $ if and only if  $\widehat{\bm{\alpha}} \in \Gamma$. We obtain the following result on prediction consistency.

\begin{theorem} \label{thm:consist}
Under Assumptions \ref{ass:main} --  \ref{ass:geometry},  as  $n \rightarrow \infty,$
$$ 
\Delta (\widehat{\bm{\alpha}}, \Gamma) \rightarrow^p 0. 
$$
\end{theorem}


\section{NUMERICAL STUDIES}\label{sec:app}

\subsection{General Considerations}

We apply geodesic optimal transport regression for two real datasets. For comparison purposes we deploy  a  Nadaraya-Watson type kernel estimator that is applicable for data in the geodesic metric space 
$ (\mathcal{M}, d)$. Related kernel estimators for object to object regression when the random objects are located on Riemannian manifolds have been considered previously  \cp{stei:09,stei:10}.  The  generalized Nadaraya-Watson estimate of responses $Y$ given a single  predictor $X$ and a sample $(X_i,Y_i)$ is 
\begin{align*}
    \widehat{Y} = \argmin_{\omega \in \mathcal{M}} \frac{\sum_{i=1}^n K(X_i, X) d^2(Y_i, \omega) }{ \sum_{i=1}^n K(X_i, X) },
\end{align*}
where $K$ is the weighting function 
$
K(\omega_1, \omega_2) = e^{- \frac{d^2(\omega_1, \omega_2)}{\tau}}
$ and $\tau$ a scaling parameter. In our empirical studies,  $\tau$ is selected by employing  a  ``median heuristic", where $\tau$ is selected as 
 $\tau = \text{median} \{ d(X_i, X) : i=1,2, \dots, n \}$. 
 
 We note that this generalized Nadaraya-Watson  regression estimator accommodates only  one predictor, whereas a key feature of GOT regression  is that  it is designed for multiple predictors. To the best of our knowledge, GOT regression is currently the only model available in the literature that operates in a general geodesic metric space and accommodates multiple predictors.

\subsection{Mortality data}

There has been long-standing interest in  studying all aspects of human mortality in the areas of demography and  the science of aging and longevity \citep{vaup:10, chiou2009modeling}. Details about the study of human longevity  can be found at \url{https://www.mortality.org/}, where   age-at-death distributions of 
males and females for 34 countries have been made available in the form of lifetables, from which one can obtain estimated continuous distributions and densities of 
age-at-at death. We consider  the age-at-death distributions separately for  males and females observed in the year 2000 as the two predictors to predict 
the age-at-death distributions for females in the year 2010 with the proposed GOT regression for two predictors.  
For these one-dimensional random distributions we use the Wasserstein metric
and place them into the Wasserstein space.  Since the comparison method 
obtained with the generalized Nadaraya-Watson esimator can only accommodate one predictor, we apply this method to predict the distribution of 
age-at death for females in 2010 using exclusively the distribution of females in 2000. 

We implemented a leave-one-out scheme to compare the two regression models. We use the distributions of age-at-death  observed  in 2000 for all  
countries excluding the $i$th country as training data to train the proposed GOT regression  and the generalized Nadaraya-Watson regression and then apply the fitted models to 
predict the $i$th country's age-at-death distribution $\mu_i$ for females in 2010. The leave-one-out error, i.e.,  the averaged Wasserstein distances $ \sum_{i=1}^{34} d_{\mathcal{W}} ( \mu_i, \widehat{\mu}_i )/34 $, for the GOT model  and the generalized Nadaraya-Watson approach were found to be 0.58 and 1.37, respectively. 

To demonstrate the comparison between GOT regression and generalized Nadaraya-Watson regression,  
we show the leave-one-out predictions for Ukraine and Latvia  in Figure \ref{fig:mort}, where  the age-at-death distribution of females is the primary predictor and that 
of males is the secondary predictor. The  fitted parameters for the GOT model are  $(\widehat{\alpha}_1, \widehat{\alpha}_2) =(0.846, 0.103)$ and $(\widehat{\alpha}_1, \widehat{\alpha}_2) =(0.851, 0.103)$ for Ukraine and Latvia respectively, 
where the first parameter $\widehat{\alpha}_1$ is the multiplier to the transport of the barycenter of the female age-at-death distributions of all countries to the female age-at-death 
distribution of the specific country in the year 2000 and the second parameter $\widehat{\alpha}_2$ is the analogously defined multiplier for males. The fact that both multipliers are smaller than one indicates the best prediction 
in the framework of this model involves a regression to the mean, as the fitted parameters and thus factors are smaller than 1, especially the factor  for males. This means the best predictions for the transports from the mean and thus the prediction of the female age-at-death distribution in 2010 is an attenuated version of the transports observed for the year 2000. For the example of Ukraine this means the prediction entails a reduction in the  left shift (increased mortality as the age-at-death distribution is to the left of the barycenter) from the barycenter distribution (Fr\'echet mean of all countries) in 2010 than what is observed for 2000, and this attenuation is very strong for the distributions corresponding to males so that the effect of the male age-at-death distribution for males is found to be small.

\begin{figure}[!t]
	\centering
	\begin{tikzpicture}[scale = 0.6]
	\node (p1) at (0,0) {\includegraphics[width=.48\textwidth, height=.48\textwidth]{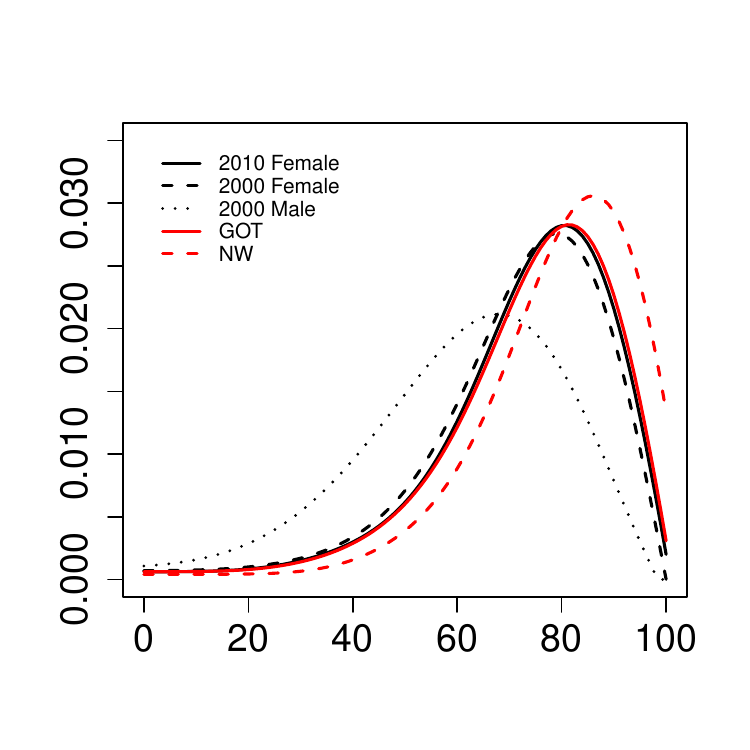}};
	\node (t1) at (0,5.5) {Ukraine};
	\node[anchor=south west] (p2) at (p1.south east) {\includegraphics[width=.48\textwidth, height=.48\textwidth]{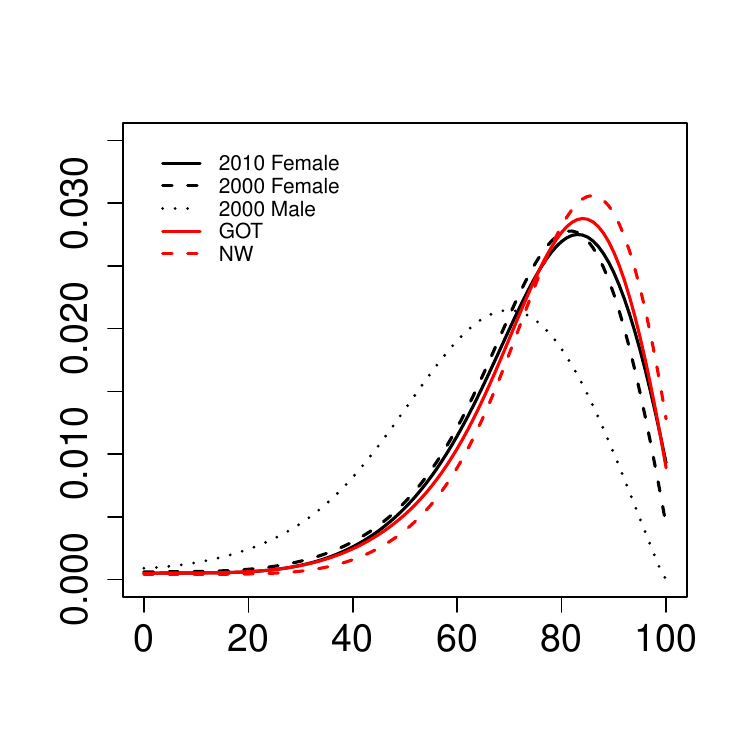}};
	\node (t2) at (13.5,5.5) {Latvia};
	\end{tikzpicture}
	\caption{GOT regression for age-at-death distributions. The solid black curve is the observed density of the age-at-death distribution of females in 2010 for Ukraine (left) and Latvia (right),  which we aim to predict from age-at-death distributions of females (dashed black) and males (dotted black) observed in 2000. We demonstrate both the density of age-at-death as predicted by GOT regression (solid red) and the density as predicted by the generalized Nadaraya-Watson method (red dashed). 
 }
	\label{fig:mort}
\end{figure}

\subsection{Temperature data} \label{app:temperature}

As a second application, we demonstrate  GOT regression for  temperature data recorded at the weather stations of airports in the U.S.. Airport stations usually have the most accurate and reliable temperature recordings. The daily minimum and maximum temperatures  can be downloaded from \url{https://www.ncdc.noaa.gov/cdo-web/search?datasetid=GHCND}. To demonstrate GOT regression on the Hilbert sphere, we consider two-dimensional density functions of daily maximum and minimum temperatures separately for each year's summer months (June - September) and winter months (December - March) for  weather stations located at 50 major airports across the U.S.. 

Here we use use the temperature profiles obtained in 2001 and 2006 to predict temperatures in 2011. More specifically, the bivariate distributions of daily minimum and maximum temperatures in the summer months of 2001 and 2006 are used as the two predictors to predict the distribution of temperatures in the summer months in 2011. For the generalized Nadaraya-Watson  estimator  the distribution  of the temperatures in the summer months of 2006 
is used to predict the bivariate distribution of temperatures in the  months of 2011. 
For the space of bivariate distributions, in which  the two predictors and the response reside,  we adopt the Fisher-Rao metric \eqref{fr} and 
thus the geodesic optimal transports of these bivariate distributions proceed along the geodesics of the Hilbert 
sphere. 

The leave-one-out prediction errors of the proposed GOT regression and of the generalized Nadaraya-Watson approach  with respect to the Fisher-Rao (intrinsic) distance  were found to be 0.19 and 0.42 respectively, demonstrating the superiority of the GOT regression for this application.  As an illustrative example, the leave-one-out predictions for the  Sky Harbor International Airport in Phoenix, AZ,  are shown in Figure \ref{fig:phe}, where also the parameters of the fitted model are reported. This fit shows that the 2006 temperature distribution is the main predictor for the 2011 temperature distribution.  
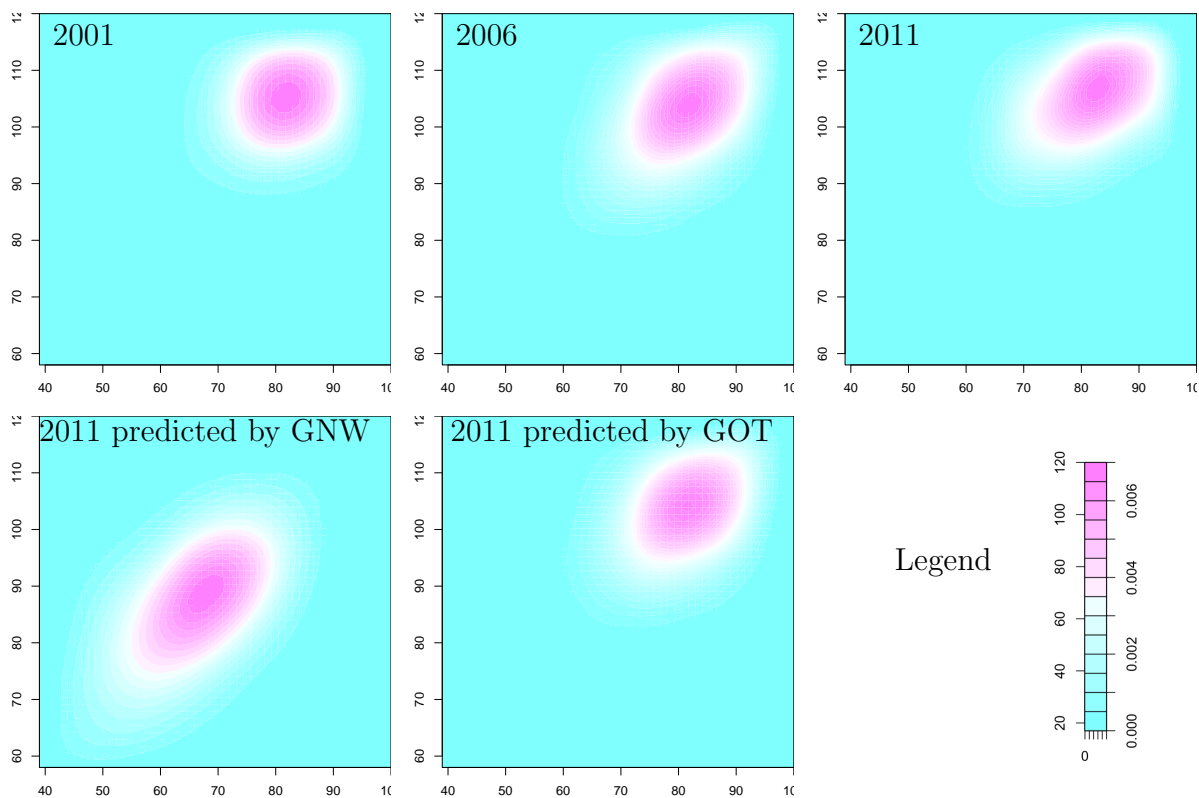
\begin{figure}[!t] \vspace{-1.5cm}
    \centering
    \begin{tikzpicture}
\matrix (m) [row sep = 0em, column sep = 0em]{    
	 \node (p11) {\includegraphics[scale=0.4]{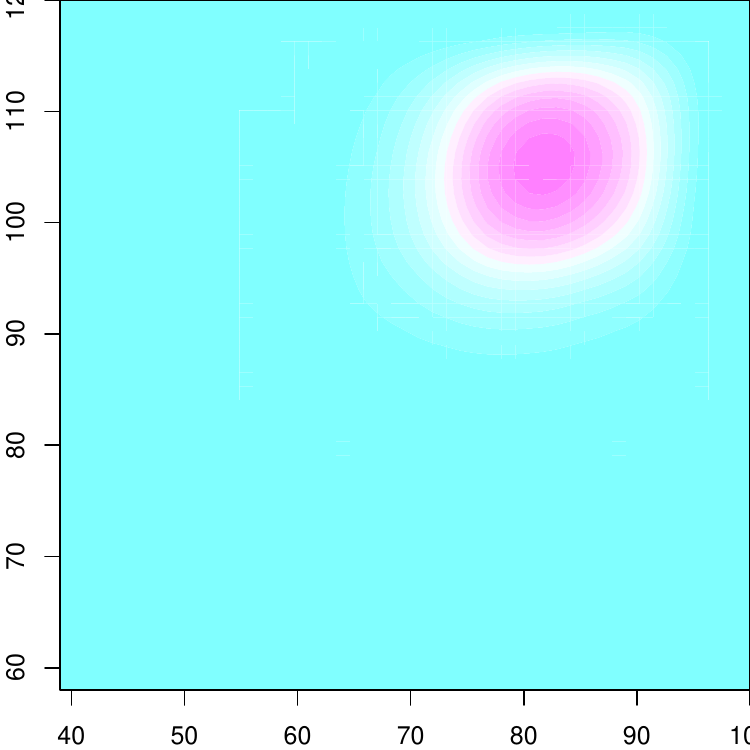}}; 
	 \node[above left = -0.7cm and -1.7cm of p11] (t12) {2001}; &
	 \node (p12) {\includegraphics[scale=0.4]{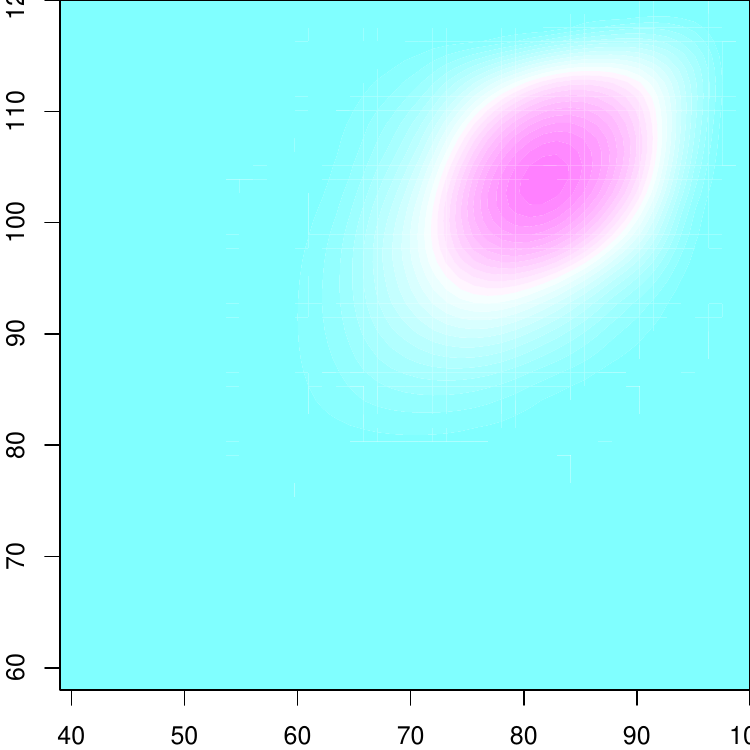}};
	 \node[above left = -0.7cm and -1.7cm of p12] (t21) {2006}; & 
  \node (p21) {\includegraphics[scale=0.4]{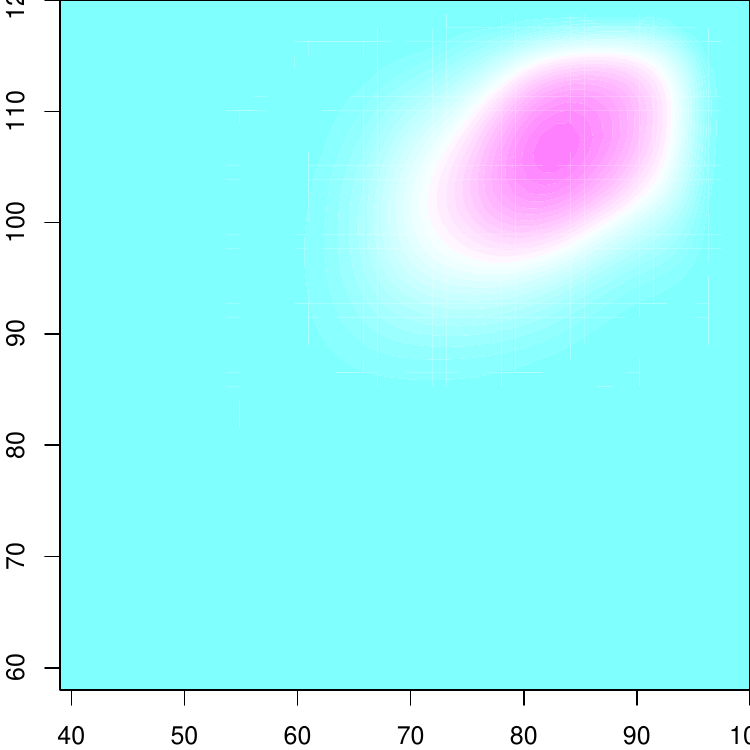}};
	 \node[above left = -0.7cm and -1.7cm of p21] (t21) {2011}; 
	 \\ 
	 \node (p31) {\includegraphics[scale=0.4]{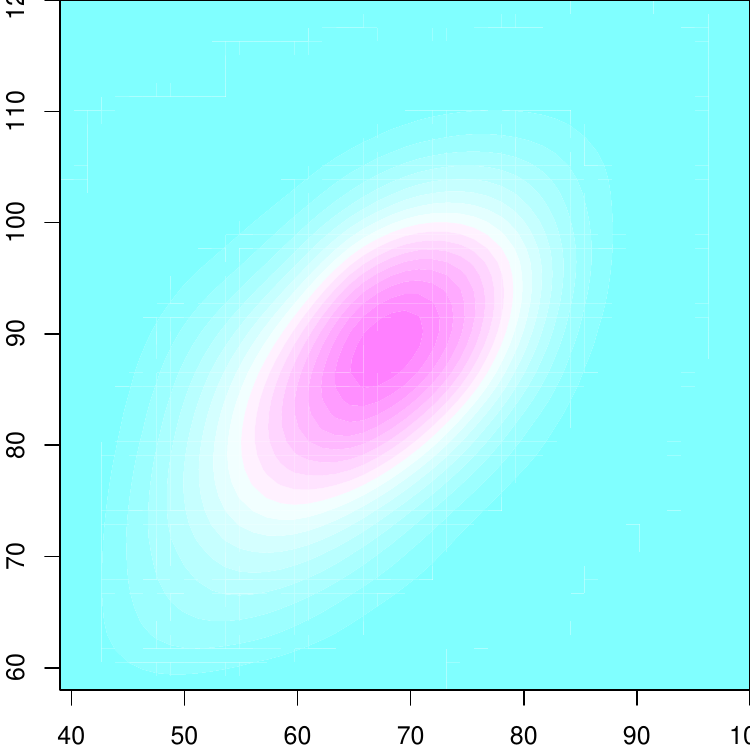}};
	 \node[above left = -0.7cm and -5.1cm of p31] (t31) {2011 predicted by GNW};
	 & \node (p32) {\includegraphics[scale=0.4]{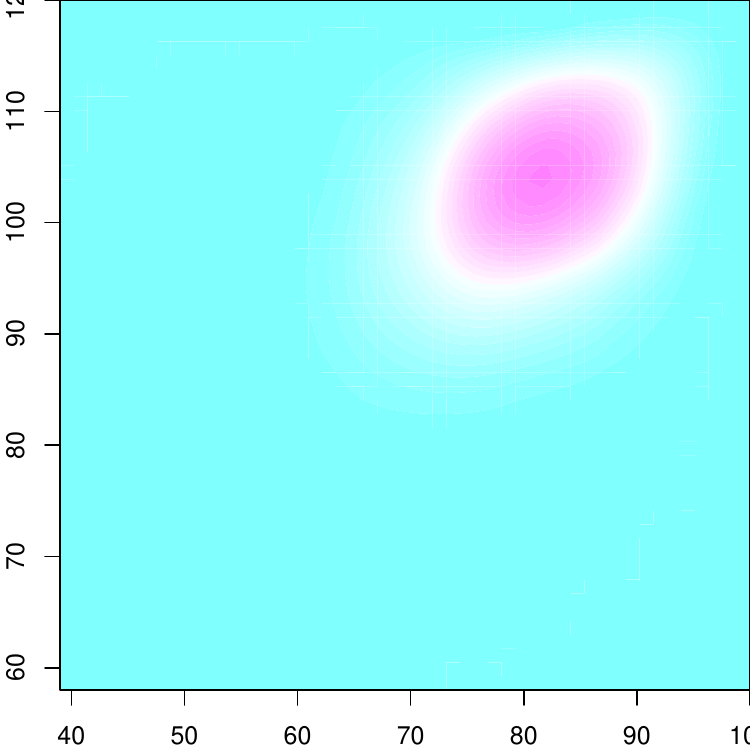}};
	 \node[above left = -0.7cm and -5.1cm of p32] (t31) {2011 predicted by GOT}; & 
	 \node (p22) at (1.2,0) {\includegraphics[scale=0.13]{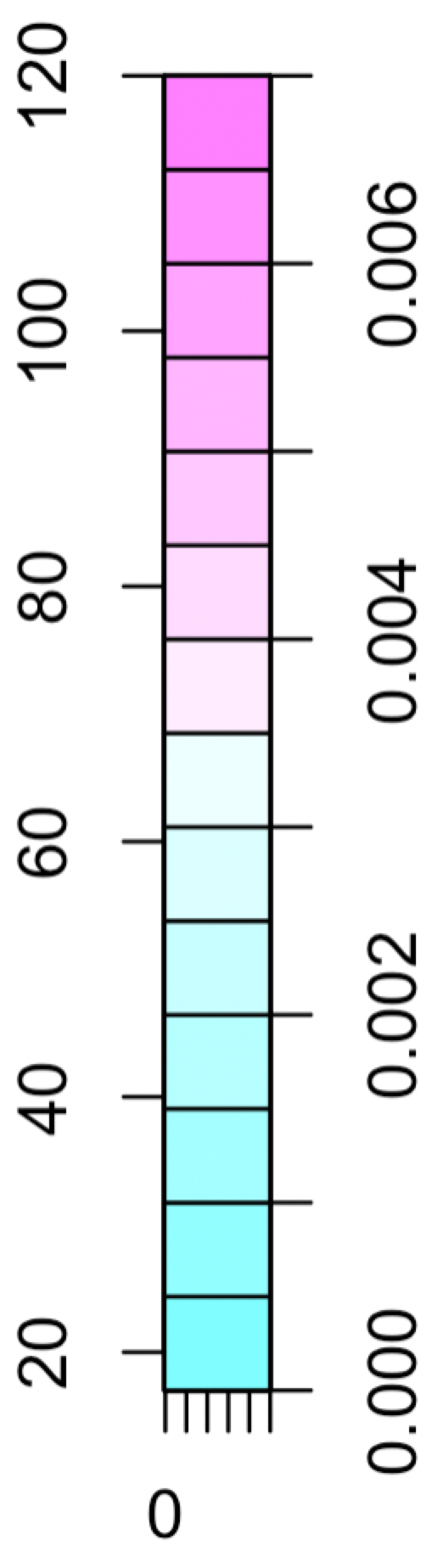}}; 
	 \node[above left = -2cm and 0.5cm of p22] (t22) {Legend}; 
	 \\
};
\node[above= -0.0cm and 0cm  of m] (title) {Phoenix Airport} ;
\end{tikzpicture}
\vspace{-.5cm}
    \caption{Contour plots of observed and predicted densities of  the joint bivariate distribution of minimum and maximum daily temperatures for the airport at Phoenix, AZ.  The first row illustrates the bivariate densities of the temperatures for summer months as  observed in 2001 (top left), 2006 (top middle) and 2011 (top right). The two two-dimensional distributions observed in 2001 and 2006 are used for GOT as the two predictors for predicting  the bivariate distribution in 2011. The predicted bivariate distributions for the daily minimum/maximum temperatures in 2011 obtained  by the GOT regression and the generalized Nadaraya-Watson kernel approach are found to have Fisher-Rao distances to the observed bivariate distribution of 0.19 and 1.1, respectively. The bivariate densities of the predicted distributions are  shown in the bottom row.
    The fitted model parameters for the GOT regression are $\widehat{\alpha}^1 = 0.97$ for the predictor corresponding to the 2006 joint temperature distribution and $\widehat{\alpha}^2 = 0.05$ for the predictor corresponding to the 2001 joint temperature distribution.}  
    \label{fig:phe}
\end{figure}

\section{DISCUSSION}\label{sec:conl}

We introduce a new interpretation of optimal transports as moving objects along geodesics in a 
metric space that features  unique geodesics and ubiquity of geodesics and refer to such transports as geodesic optimal transports. The ubiquity assumption is satisfied for 
most unique geodesic spaces that are  of interest for statistical analysis, sometimes  by adding a projection step,  and thus is only a mild limitation.

 Making use of this new concept of optimal transports, we introduce a geodesic regression model where we make use of a previously introduced transport algebra. It is the first model for random objects, i.e.,  data in metric spaces, that covers vectors of objects predictors and it allows for interpretation of the effects of individual predictors. While we provide consistency results under additional regularity assumptions,  rates of convergence will require both further assumptions and new theoretical 
 concepts and to obtain such rates  will be left for future research.
 
 The proposed GOT regression is seen to work well 
 for multiple random object predictors paired with a random object response when the objects are situated in geodesic spaces that satisfy the requirements. The GOT 
 model proves useful for many scenarios with complex data 
 types as encountered increasingly in contemporary data analysis. We expect that ransport based regression approaches will stimulate further research
 on modeling regression relations in the emerging field of  random objects and metric statistics.

{
\bibliographystyle{agsm}
\bibliography{Functional}
}

\end{document}